\documentclass[letterpaper, 10 pt, conference]{ieeeconf}  

\IEEEoverridecommandlockouts                              

\usepackage[utf8]{inputenc}
\usepackage[T1]{fontenc}
\usepackage{amsmath}

\usepackage{amsthm}
\usepackage{amsfonts}
\usepackage{amssymb}
\usepackage{graphicx}
\usepackage[outdir=./]{epstopdf}
\usepackage{xcolor}
\usepackage[hyphens]{url}
\usepackage{hyperref}
\usepackage[hyphenbreaks]{breakurl}
\usepackage{cite}
\usepackage{booktabs}
\usepackage{siunitx}
\usepackage[font=footnotesize]{caption}
\usepackage{subcaption}

\usepackage{tikz}

\newtheorem{proposition}{Proposition}
\newtheorem{problem}{Problem}
\newtheorem{corollary}{Corollary}
\newtheorem{remark}{Remark}
\newtheorem{theorem}{Theorem}
\newtheorem{assumption}{Assumption}

\newcommand{\xnom}{x_{\text{n}}}
\newcommand{\unom}{u_{\text{n}}}
\newcommand{\ynom}{y_{\text{n}}}
\newcommand{\mrpi}{\mathbb{Z}_K}
\newcommand{\X}{\mathbb{X}}
\newcommand{\U}{\mathbb{U}}
\newcommand{\Xc}{\mathbb{X}_\text{c}}
\newcommand{\Uc}{\mathbb{U}_\text{c}}
\newcommand{\Xa}{\mathbb{X}_{\text{a},\bar{K}}}
\newcommand{\Xf}{X_{f,\bar{K}}^{\lambda}}
\newcommand{\RMPC}{R-MPC}
\newcommand{\RTMPC}{RT-MPC}
\newcommand{\ERTMPC}{ERT-MPC}

\newcommand\copyrighttext{%
  \footnotesize © 2024 IEEE.  Personal use of this material is permitted. Permission from IEEE must be obtained for all other uses, in any current or future media, including reprinting/republishing this material for advertising or promotional purposes, creating new collective works, for resale or redistribution to servers or lists, or reuse of any copyrighted component of this work in other works.}
\newcommand\copyrightnotice{%
\begin{tikzpicture}[overlay, remember picture]
\node[anchor=south,yshift=10pt] at (current page.south) {\fbox{\parbox{\dimexpr\textwidth-\fboxsep-\fboxrule\relax}{\copyrighttext}}};
\end{tikzpicture}%
}

\title{\LARGE \bf
Remote Tube-based MPC for Tracking Over Lossy Networks
}

\author{David Umsonst$^{\dagger}$ and Fernando S. Barbosa$^{\dagger}$
\thanks{\hspace{-1.05em}\footnotesize \textsuperscript{$\dagger$} Ericsson Research, Stockholm, Sweden.\newline
        {\tt\scriptsize\{david.umsonst, fernando.dos.santos.barbosa\}@ericsson.com}}%
}

\begin{document}

\maketitle
\thispagestyle{empty}
\pagestyle{empty}

\begin{abstract}
This paper addresses the problem of controlling constrained systems subject to disturbances in the case where controller and system are connected over a lossy network.
To do so, we propose a novel framework that splits the concept of tube-based model predictive control into two parts. 
One runs locally on the system and is responsible for disturbance rejection, while the other runs remotely and provides optimal input trajectories that satisfy the system's state and input constraints.
Key to our approach is the presence of a nominal model and an ancillary controller on the local system.
Theoretical guarantees regarding the recursive feasibility and the tracking capabilities in the presence of disturbances and packet losses in both directions are provided.
To test the efficacy of the proposed approach, we compare it to a state-of-the-art solution in the case of controlling a cartpole system. 
Extensive simulations are carried out with both linearized and nonlinear system dynamics, as well as different packet loss probabilities and disturbances.
The code for this work is available at \url{https://github.com/EricssonResearch/Robust-Tracking-MPC-over-Lossy-Networks}
\end{abstract}

\copyrightnotice
\section{Introduction}
\label{sec:introduction}

Wireless communication has evolved to enable higher and faster data transfer, with 5G being envisioned as being a key enabler of Industry 4.0 \cite{5g-acia_key_2020, 5g-smart_2022} and of mass digitalization.
Looking into control systems and robotics in general, faster and more reliable wireless communication enables plants and systems to be controlled remotely, utilizing edge and cloud computing, in a so-called offloaded control \cite{baxi_towards_2022}. 
Running heavy-processing components remotely allows industries to save costs with cabling and processing power in the plant, easier integration of autonomous mobile agents in the industrial floor, and also a reduced energy consumption on battery-powered agents.

However, any wireless network is subject to imperfections and constraints. 
The former means that it can present delays, packet drops, and even longer outages.
The latter implies that its resources, such as throughput and load, are constrained. 
These two factors are specially precarious for time- and safey-critical systems, such as unstable plants, mobile robots and autonomous cars \cite{WirelessNetworksSurvey}.

A popular approach to address the problem of stabilization under safety and actuator constraints is Model Predictive Control (MPC) \cite{MPCBook}, since such constraints can be explicitly accounted for in its formulation. 
Several approaches have been proposed to make MPC robust to network imperfection.
Looking specifically into the stabilization problem, \cite{SelfTriggeredMPCBoundedPacketLoss} considers a bound on the amount of consecutively lost packets, while \cite{NetworkedMPCWithUncertaintyAndBoundedDelay} considers bounded delay.
Moving to trajectory tracking problems, \cite{ReferenceTrackingWithStochasticMPCOverNetwork} assumes Bernoulli distributed packet loss, while \cite{RemoteMPCoverLossyNetworks} only assumes that from time to time there are consecutive successful packet deliveries from the plant to the controller and back.
In addition to network imperfections, \cite{NetworkedMPCWithUncertaintyAndBoundedDelay} considers a bounded disturbance and \cite{ReferenceTrackingWithStochasticMPCOverNetwork} considers an unbounded zero mean stochastic disturbance acting on the plant.

Extensive research has been carried out on MPC that disregards the effects of imperfect communication, either because the controller is running onboard or because perfect communication was assumed, but can handle local disturbances.
Limon \emph{et al.} \cite{RobustTubeBasedTrackingMPC} propose a robust tracking MPC that keeps the plant state in a bounded neighborhood of the nominal plant state, while tracking a constant reference. Here, the nominal plant represents the plant dynamics without a disturbance present.
Roque \emph{et al.}  \cite{CorridorMPC} combine control barrier function with the nominal system to guarantee that the continuous system is within a bounded neighborhood of the desired reference in between discrete controller updates.
Neither \cite{RobustTubeBasedTrackingMPC} nor \cite{CorridorMPC} can handle network imperfections.

In our work, we combine the mild network assumptions of \cite{RemoteMPCoverLossyNetworks} with the disturbance rejection of \cite{RobustTubeBasedTrackingMPC} to develop a novel remote tracking MPC framework.
\emph{This framework guarantees the satisfaction of state and actuator constraints in the presence of a local disturbance and a lossy network.} 
The key idea is to use a nominal model on the local plant to simulate the nominal plant state in case of packet losses. 
This nominal plant state allows us to reduce the bandwidth by sending only control input trajectories over the network and it is used in an ancillary controller to reject the disturbance.
This allows us to handle both packet losses and local disturbances.
Furthermore, the code for our approach is available online.\footnote{\url{https://github.com/EricssonResearch/Robust-Tracking-MPC-over-Lossy-Networks}}

\textit{Notation:}
Let $x\in\mathbb{R}^n$ and $A\in\mathbb{R}^{n\times m}$ be a real-valued $n$-dimensional column vector and matrix with $n$ rows and $m$ columns, respectively. 
The transpose of a vector $x$ and matrix $A$ are $x^\top$ and $A^\top$, respectively.
The spectral radius and matrix square root of a square matrix $A$ are denoted by $\rho(A)$ and  $A^{\frac{1}{2}}$, respectively.
The $n$ dimensional identity matrix is denoted by $I_n$, while $0$ denotes a scalar, vector, or matrix with zero elements of appropriate dimensions.
A symmetric and square positive (semi-)definite matrix $A$ is denoted by $A>0(A\geq 0)$ and we use $\|x\|_A^2=x^\top A x$.
For a set $\mathbb{P}$ and a matrix $A$ of appropriate dimension, we define $A\mathbb{P}=\lbrace A p\ |\ p\in\mathbb{P} \rbrace$. 
For two sets $\mathbb{P}$ and $\mathbb{Q}$, the Minkowski sum and the Pontryagin difference are denoted as $\mathbb{P} \oplus \mathbb{Q}$ and $\mathbb{P} \ominus \mathbb{Q}$, respectively.
The probability of an event $E$ is denoted by $\mathrm{Prob}(E)$.

\section{Problem Definition}
\label{sec:background}

A block diagram summarizing the components involved in our setup is presented in Figure~\ref{fig:Overview}. In what follows in this section, we describe such components and formulate the problem addressed in this paper.

\begin{figure}
    \vspace{2mm}
    \centering
    \includegraphics[trim=155 20 240 10,clip,width=0.7\linewidth]{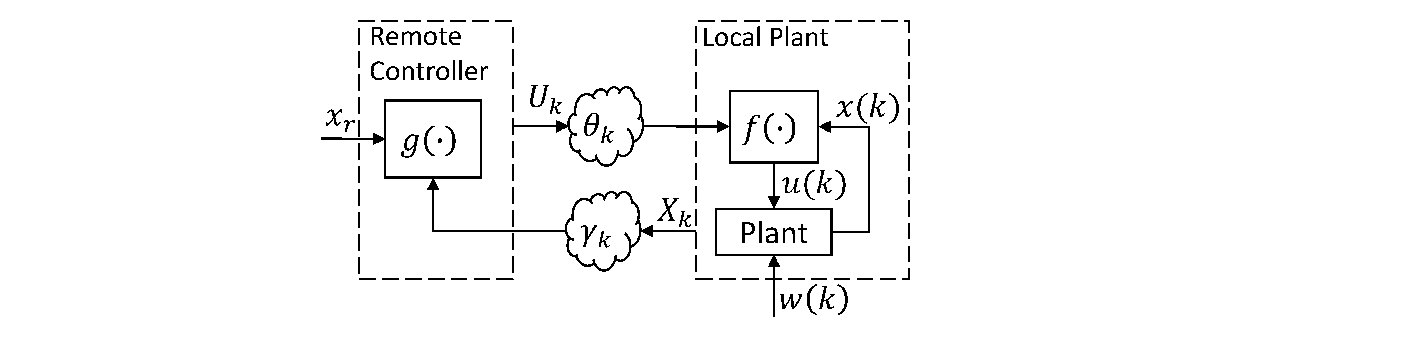}
    \caption{Block diagram of the problem setup}
    \label{fig:Overview}
\end{figure}

\subsubsection{Network} 
Local plant and remote controller communicate via a potentially \emph{lossy network}, in which network packets can be lost in both directions. Reasons for a lost packet include a large transmission delay, a packet drop in the network, reordering, or a short network outage.
To model these packet losses, we introduce two variables: $\theta_k$ and $\gamma_k$.
The variable $\theta_k\in\lbrace 0,1 \rbrace$ indicates whether the local plant has received the packet $U_k$ or not, i.e., $\theta_k=1$ if $U_k$, sent from the remote controller at time step $k$, has been received at the local plant, and $\theta_k=0$ otherwise.
Similarly, the variable $\gamma_k\in\lbrace 0,1 \rbrace$ indicates that the packet $X_k$ sent from the local plant has been received at the remote controller ($\gamma_k=1$) or not ($\gamma_k~=~0$).

\begin{assumption}
\label{assum:Network}
Over time, there is an infinite amount of two successful consecutive transmissions from plant to controller and controller to plant, i.e.,
\begin{equation}
    \mathrm{Prob}(\cap_{t\geq k}\lbrace \gamma_{t-1}\theta_t=0\rbrace) = 0\ \forall\ k\geq 0.
\end{equation}
\end{assumption}
This assumption is as in \cite{RemoteMPCoverLossyNetworks}, and does not put any major restrictions on the reasons for the packet loss, such as a fixed distribution or a maximum amount of lost packets in a row.

\subsubsection{Local plant}
Consider a linear time-invariant discrete-time plant with additive disturbance given by
\begin{equation}
    \label{eq:Plant}
    \begin{aligned}
        x(k+1) &= Ax(k) + Bu(k) + w(k),\\
        y(k) &= Cx(k),
    \end{aligned}
\end{equation}
where $x(k)\in\mathbb{R}^{n_x}$, $u(k) \in \mathbb{R}^{n_u}$, $y(k)\in\mathbb{R}^{n_y}$ and ${w(k)\in\mathbb{R}^{n_x}}$ are the plant's state, control input, output, and disturbance at time step $k\in\mathbb{N}_{\geq 0}$, respectively.
Here, ${A\in\mathbb{R}^{n_x \times n_x}}$, ${B\in\mathbb{R}^{n_x \times n_u}}$, and $C\in\mathbb{R}^{n_y \times n_x}$ are the system, input, and output matrices, respectively.

\begin{assumption}
\label{assum:StabilizableSystem}
The system $(A,B)$ is stabilizable.
\end{assumption}
This assumption is necessary to be able to design a controller that stabilizes the plant \eqref{eq:Plant}.

\begin{assumption}
    \label{assum:PolytopeDisturbance}
    The disturbance is bounded by a compact set $\mathbb{W}$, such that $w(k)\in\mathbb{W}$ for all $k$, where 
    \begin{equation}
        \mathbb{W} = \lbrace w\in\mathbb{R}^{n_x}\ |\ H_w w \leq h_w \rbrace,
    \end{equation}
    and $\mathbb{W}$ contains the origin in its interior.
\end{assumption}
This assumption confines the disturbance to a bounded set, which could, for example, depend on the modelling errors.

Furthermore, we also consider constraints in state $x(k)\in \mathbb{X}$ and input $u(k)\in \mathbb{U}$. These sets indicate, for example, safe set of states in which the plant should evolve, and actuator saturation. 
If $x(k)\in \mathbb{X}$ and $u(k)\in \mathbb{U}$, $x(k)$ and $u(k)$ are called \emph{admissible}.
\begin{assumption}
    \label{assum:StateAndInputSet}
    The sets $\X$ and $\U$ are bounded sets containing the origin in their interior and are defined as
    \begin{align}
        \X &= \lbrace x\in\mathbb{R}^{n_x}\ |\ H_x x \leq h_x \rbrace,\\
        \U &= \lbrace u\in\mathbb{R}^{n_u}\ |\ H_u u \leq h_u \rbrace.
    \end{align}
\end{assumption}

The control input is determined as ${u(k)=f(x(k),\lbrace U_i\rbrace_{i=0}^k, \lbrace\theta_i\rbrace_{i=0}^k)}$, where $\lbrace U_i\rbrace_{i=0}^k$ and $\lbrace\theta_i\rbrace_{i=0}^k$ are the sequence of packets sent from the remote controller to the local plant and the binary sequence indicating the successful transmission of them, respectively. Note that this function can make use of all previously received packets.

\subsubsection{Remote Controller}
The remote controller is used to determine the controller packet $U_k$ based on the received packets $X_i$ and the desired reference $x_r$.
More formally the controller is defined as $g(\lbrace U_i\rbrace_{i=0}^{k-1},\lbrace X_i\rbrace_{i=0}^{k-1}, \lbrace\gamma_i\rbrace_{i=0}^{k-1},x_r)$, which has access to all previous controller packets and can make use of all previously received plant packets. Here, $\lbrace X_i\rbrace_{i=0}^{k-1}$ and $ \lbrace\gamma_i\rbrace_{i=0}^{k-1}$ are defined similarly as $\lbrace U_i\rbrace_{i=0}^k$ and $\lbrace\theta_i\rbrace_{i=0}^k$ above.

\subsubsection{Problem Formulation} 
Now that all the components are defined, let us formulate the problem we want to solve.
\begin{problem}
\label{prob:TrackingOverLossyNetworkWithLocalDisturbance}
Given a local plant \eqref{eq:Plant}, design $f(\cdot)$ and $g(\cdot)$ such that i) state and input constraints are respected, i.e. $x(k)\in\mathbb{X}$ and $u(k)\in\mathbb{U}$ for $k\geq 0$, and ii) $x(k)$ converges to a bounded neighborhood of reference $x_r(k) \in \mathbb{R}^{n_x}$, despite the lossy network and the disturbance $w(k)$.
\end{problem}

\section{Preliminaries}
\label{sec:Preliminaries}

In the previous section, we have set up our problem and and now we will present several preliminaries, found, e.g., in \cite{RemoteMPCoverLossyNetworks, RobustTubeBasedTrackingMPC, RobustTubeBasedRegulatorMPC}, necessary for our proposed approach.
This section introduces the \emph{nominal plant} dynamics, i.e. the plant dynamics without an additive disturbance, the error between the actual and the nominal plant state, as well as the steady-state behaviour of the nominal plant.

\subsubsection{Nominal plant} 
The nominal plant \cite{RobustTubeBasedRegulatorMPC} is given by 
\begin{equation}
    \label{eq:NominalPlant}
    \begin{aligned}
        \xnom(k+1) &= A\xnom(k) + B\unom(k),\\
        \ynom(k) &= C\xnom(k),
    \end{aligned}
\end{equation}
where $\xnom(k)\in\mathbb{R}^{n_x}$, $\unom(k)\in \mathbb{R}^{n_u}$, and  $\ynom(k)\in \mathbb{R}^{n_y}$ are the nominal state, the nominal control input, and the nominal output, respectively.
Due to the disturbance $w(k)$ in \eqref{eq:Plant}, the plant state differs from the nominal state and subsequently we want to show how close the plant state is to the nominal state. To do so, we introduce the error $e(k)=x(k)-\xnom(k)$.

If $A$ is unstable, then the error will diverge such that the plant state is not close to the nominal state. To prevent that, we introduce an \emph{ancillary controller}, which will be used by the plant to track the nominal state. The ancillary controller is given by
\begin{equation}
    \label{eq:AncillaryController}
    u(k) = \unom(k) - K\left(x(k)-\xnom(k)\right),
\end{equation}
where $K\in\mathbb{R}^{n_u\times n_x}$ is a linear state feedback controller chosen such that $\rho(A-BK)<1$, which is possible due to Assumption~\ref{assum:StabilizableSystem}. Note that if the system matrix $A$ is stable, i.e., $\rho(A)<1$, then we could choose $K=0$.

When the plant uses the ancillary controller \eqref{eq:AncillaryController}, we obtain the following error dynamics
\begin{equation}
    e(k+1) = (A-BK)e(k)+w(k).
\end{equation}
The evolution of $e(k)$ is bounded, because $\mathbb{W}$ is a compact set and $A-BK$ is stable \cite{OutputAdmissibleSet}.

We introduce the minimal robust positively invariant set \cite{ApproxmiationOfMininamlRobustPositivelyInvariantSet} to determine the bounded set in which  $e(k)$ evolves as
\begin{equation}
    \mrpi = \bigoplus_{i=0}^{\infty}(A-BK)^i\mathbb{W}.
\end{equation}
It is guaranteed that $(A-BK)\mrpi\oplus\mathbb{W}\subseteq\mrpi$, i.e., if ${e(k_0)\in\mrpi}$, then $e(k)\in\mrpi$ for all $k>k_0$.
Since $0\in \mathbb{W}$, we have $0\in\mrpi$ \cite{ApproxmiationOfMininamlRobustPositivelyInvariantSet}.
The set $\mrpi$ can be overapproximated with, for example, the methods proposed in \cite{ApproxmiationOfMininamlRobustPositivelyInvariantSet} and \cite{RPI_Darup}.

With $\mrpi$ defined, it is known that \cite{RobustTubeBasedRegulatorMPC} 
\begin{equation}
    x(k) \in \lbrace \xnom(k)\rbrace \oplus \mrpi\ \forall k>0,
\end{equation}
given that $x(0)\in \lbrace \xnom(0)\rbrace \oplus \mrpi$
This means that the plant state evolves in a bounded neighborhood $\mrpi$ around the nominal state. This bounded neighborhood is often called a \emph{tube}.
The size of $\mrpi$ depends on the ancillary controller $K$, so that the ancillary controller determines how close the plant state will track the nominal state.
Similarly, we obtain
\begin{equation}
u(k) \in \lbrace \unom(k)\rbrace \oplus (-K)\mrpi,
\end{equation}
which means that the control input also evolves in a bounded neighborhood around the nominal control input.

Therefore, we will introduce tightened constraint sets \cite{RobustTubeBasedRegulatorMPC} in which the nominal state and input trajectory should evolve, i.e., ${\xnom(k)\in\Xc}$ and ${\unom(k)\in\Uc}$, which guarantee that the plant state and input trajectories evolve in the sets $\X$ and $\U$, respectively.
We define the \emph{tightened sets} $\Xc = \X \ominus \mrpi$ and ${\Uc = \U \ominus (-K)\mrpi}$, which guarantee that $\Xc \oplus \mrpi \subseteq \X$ and $\Uc \oplus (-K)\mrpi \subseteq \U$.

\subsubsection{Steady-state behavior} 
Next, we look into the steady-states of the nominal plant \cite{RobustTubeBasedTrackingMPC, RemoteMPCoverLossyNetworks} and how to control the nominal plant towards a steady state while guaranteeing that the nominal state and input remain in $\Xc$ and $\Uc$, respectively.

The steady-state equations of \eqref{eq:NominalPlant} are given by
\begin{align}
\label{eq:SteadyStateEquation}
\begin{bmatrix}
A-I_{n_x} & B
\end{bmatrix}\begin{bmatrix}
\bar{x} \\ \bar{u} 
\end{bmatrix}=0,
\end{align} 
which have a solution due to Assumption~\ref{assum:StabilizableSystem}. Here, $\bar{x}~\in~\mathbb{R}^{n_x}$ and $\bar{u}\in\mathbb{R}^{n_u}$ are a steady state and steady-state input, respectively.
To control the nominal system towards the steady state, we introduce the state feedback controller ${\bar{K}\in\mathbb{R}^{n_x \times n_u}}$ for the nominal plant
\begin{equation}
\label{eq:nominalsteadystatecontroller}
\unom(k) = \bar{u} - \bar{K}(\xnom(k) - \bar{x}),
\end{equation}
where $\bar{K}$ is chosen such that $\rho(A-B\bar{K})<1$.
However, we want to guarantee that $\xnom\in\Xc$ and $\unom\in\Uc$. 
Thus, we define the augmented state ${x_a(k)=[\xnom^\top(k),\ \bar{x}^\top,\ \bar{u}^\top]^\top}$ and its dynamics with the controller in \eqref{eq:nominalsteadystatecontroller} are given by
\begin{align*}
x_\text{a}(k+1)= A_\text{a} x_\text{a}(k)\ \mathrm{with}\ A_\text{a}=\begin{bmatrix}
A-B\bar{K} & B\bar{K} & B\\
0 & I_{n_x} & 0\\
0 & 0 & I_{n_u}
\end{bmatrix}.
\end{align*}
Next, we define the maximum admissible set \cite{OutputAdmissibleSet}
\begin{equation}
X_{f,\bar{K}} = \lbrace x_\text{a}\ |\ A_\text{a}^k x_\text{a} \in \Xa\ \forall\,k\in\mathbb{N}_{\geq 0}\rbrace,
\end{equation}
where ${\Xa = \lbrace x_\text{a}| \xnom\in\Xc, \bar{u} - \bar{K}(\xnom(k) - \bar{x})\in\Uc \rbrace}$.
If ${[\xnom(0)^\top, \bar{x}^\top, \bar{u}^\top]^\top \in X_{f,\bar{K}}}$, then the nominal plant \eqref{eq:NominalPlant} using the control law \eqref{eq:nominalsteadystatecontroller} guarantees that ${[\xnom(k)^\top, \bar{x}^\top, \bar{u}^\top]^\top \in X_{f,\bar{K}}}$ for all $k>0$ and that $\xnom(k)$ converges to the steady state $\bar{x}$.
We can compute $X_{f,\bar{K}}$ as described in \cite{OutputAdmissibleSet}. 
Since $X_{f,\bar{K}}$ might not be finitely determined, i.e., the polytope $X_{f,\bar{K}}$ cannot be described by a finite amount of inequalities, we introduce
\begin{equation}
\Xf = X_{f,\bar{K}} \cap \lbrace \bar{x},\ \bar{u}\ |\ \bar{x}\in\lambda \Xc,\ \bar{u} \in \lambda\Uc \rbrace,
\end{equation}
with $\lambda\in(0,1)$. 
This is a finitely determined set that approximates $X_{f,\bar{K}}$ arbitrarily well as $\lambda\rightarrow 1$ \cite{OutputAdmissibleSet}.

\section{Remote Tube-Based Tracking MPC\\ over Lossy Networks}
\label{sec:MPCdesign}
In this section, we describe in more details the \emph{Remote Tube-based Tracking MPC over Lossy Networks} approach that we propose to solve Problem \ref{prob:TrackingOverLossyNetworkWithLocalDisturbance} and its theoretical guarantees. 
As mentioned earlier, the proposed approach is an extension of those presented in \cite{RemoteMPCoverLossyNetworks} and \cite{RobustTubeBasedTrackingMPC} that enables remote tracking of references even in the presence of disturbance on the plant and lossy networks.

Figure~\ref{fig:RemoteTubeTrackingMPCLayout} presents the architecture of our proposed approach.
It is composed of five parts: two are placed remotely representing $g(\cdot)$, namely the MPC controller and the state estimator, and three are placed together with the local plant representing $f(\cdot)$, namely the consistent actuator, the nominal plant, and the ancillary controller.

\begin{figure}
    \centering
    \includegraphics[trim=105 8 115 0,clip,width=\linewidth]{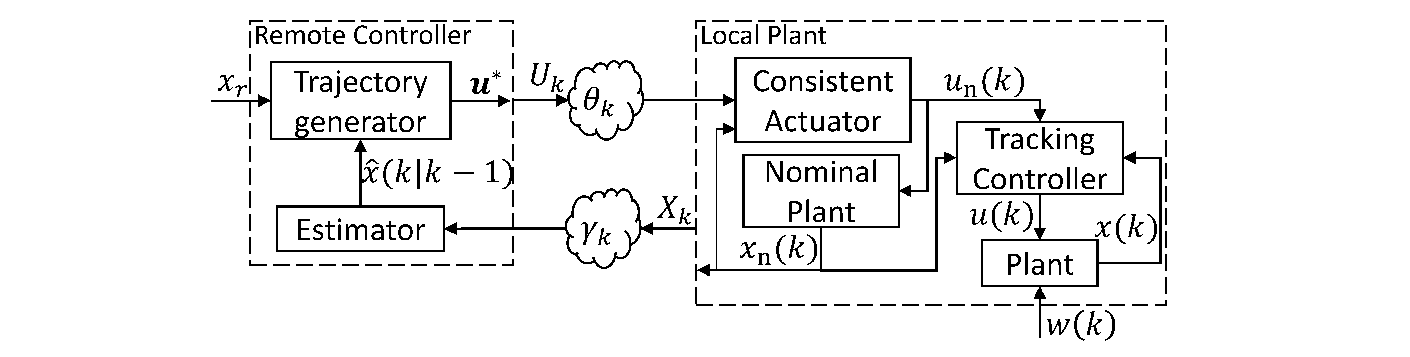}
    \caption{The block diagram of our proposed remote tube-based tracking MPC approach over lossy networks.}
    \label{fig:RemoteTubeTrackingMPCLayout}
\end{figure}

\subsection{Remote Model Predictive Controller For Tracking}
\label{sec:ProposedRemoteMPC}
To track the reference $x_r$, we will use a model predictive controller on the remote controller-side, which is inspired by \cite{RemoteMPCoverLossyNetworks}.
The cost function optimized in the MPC is given by
\begin{equation}
    c(\mathbf{u},\mathbf{x},\bar{x},\bar{u} ,x_r) = \sum_{i=0}^{N-1}\left(c_i(\mathbf{u},\mathbf{x},\bar{x},\bar{u})\right)+\bar{c}(\mathbf{x},\bar{x},x_r),
\end{equation}
where $\mathbf{u}=\lbrace\mathbf{u}(0),\ldots,\mathbf{u}(N)\rbrace$, $\mathbf{x}=\lbrace\mathbf{x}(0),\ldots,\mathbf{x}(N)\rbrace$,
\begin{align}
    c_i(\mathbf{u},\mathbf{x},\bar{x},\bar{u})&=\|\mathbf{x}(i)-\bar{x}\|_Q^2+\|\mathbf{u}(i)-\bar{u}\|_R^2,\\
    \bar{c}(\mathbf{x},\bar{x},x_r) &= \|\mathbf{x}(N)-\bar{x}\|_P^2+\|\bar{x}-x_r\|_T^2,
\end{align}
and $Q\geq 0$, $R>0$, and $T>0$ are the symmetric cost matrices for the state, input, and the tracking output, and $P$ is the solution of $P = (A-B\bar{K})^\top P(A-B\bar{K})+Q+\bar{K}^\top R \bar{K}.$
Given the cost function, a state estimate $\hat{x}(k|k-1)$ and a reference signal $x_r$, the optimization problem of the MPC is formulated as follows
\begin{subequations}
\label{eq:RemoteRobustTubeTrackingMPC}
\begin{align}
        \min_{\mathbf{u},\bar{x},\bar{u}}\ &c(\mathbf{u},\mathbf{x},\bar{x},\bar{u} ,x_r)\\
        \mathrm{s.t.}\ &\mathbf{x}(i+1) = A\mathbf{x}(i) + B\mathbf{u}(i),\\
        	 &\mathbf{x}(i) \in \Xc,\ \mathbf{u}(i) \in \Uc,\ 
        i\in\lbrace 0,\ldots,N-1\rbrace,\\
        	 &\mathbf{x}(0) = \hat{x}(k|k-1), \label{eq:MPCInitCondition}\\
        	 &(\mathbf{x}(N),\,\bar{x},\,\bar{u}) \in \Xf,\\
        	 &\begin{bmatrix}
        		A-I_{n_x} & B\\
        		\end{bmatrix}\begin{bmatrix}
        \bar{x} \\ \bar{u}
        \end{bmatrix}=0
\end{align}
\end{subequations}
where the sets $\Xc$, $\Uc$, and $\Xf$, are as in Section~\ref{sec:Preliminaries}.

Compared to the remote MPC formulated in \cite{RemoteMPCoverLossyNetworks}, the MPC \eqref{eq:RemoteRobustTubeTrackingMPC} generates trajectories for the nominal plant by using the tightened sets $\Xc$, $\Uc$, and $\Xf$. This difference is inspired by \cite{RobustTubeBasedTrackingMPC} and we make use of it in Section~\ref{sec:NominalPlantAndAncillaryController} to generate inputs $u(k)\in\U$, which guarantee $x(k)\in\X$ for all ${k\in \mathbb{N}_{\geq 0}}$.
\emph{Since the communication from the plant to the remote controller is lossy, we are not guaranteed to have the plant state $x(k)$ available at time $k$.} 
Therefore, we use the estimate ${\hat{x}(k|k-1)}$ based on the previously received packets as in \cite{RemoteMPCoverLossyNetworks} (see Section~\ref{sec:Estimator}) instead of the true state.

Let the optimal solution of \eqref{eq:RemoteRobustTubeTrackingMPC} at time step $k$ be $\mathbf{u}_k^*, \bar{u}_k^*$, and $\bar{x}_k^*$. With that, the packet $U_k$ is, similar to \cite{RemoteMPCoverLossyNetworks}, constructed as follows
\begin{equation}
    \label{eq:Uk}
    U_k = \lbrace \mathbf{u}_k^*, \bar{u}_k^*+\bar{K}\bar{x}_k^*, q_k\rbrace,
\end{equation}
where $q_k$ is the time instance when the remote estimator has last received a packet from the local plant. The packet contains the optimal nominal input trajectory at time step $k$ and the steady-state control input for the nominal plant.

\subsection{Consistent Actuator}
The consistent actuator is located at the local plant and is responsible for deciding the next nominal control input $\unom(k)$. 
It has the same functionality as the Smart Actuator in \cite{RemoteMPCoverLossyNetworks}.
When a packet $U_k$ is received, the consistent actuator needs to decide if $U_k$ will be used or if it will be discarded; in the latter, the packet already in use continues to be applied.

The consistent actuator might discard a received packet because the estimated state on the remote controller side is inconsistent with the actual state on the plant. 
This means that the control inputs have been calculated based on an incorrectly estimated state.
To determine consistency, we use a variable $\Theta_k$ as in \cite{RemoteMPCoverLossyNetworks}, which is calculated as follows
\begin{align}
    \Theta_k = \begin{cases}
        \prod_{i=q_k+1}^k\theta_i &\mathrm{if}\ \theta_k=1,\\
        0 &\mathrm{otherwise.}
    \end{cases}
\end{align}

We observe that if $\theta_k=1$, i.e., the packet is received at time step $k$, then we can calculate the product and otherwise ${\Theta_k=0}$.
Once $\Theta_k$ is determined the consistent actuator updates its internal state $s_k$ as follows
\begin{equation}
    s_k = \Theta_k k + (1-\Theta_k)s_{k-1}.
\end{equation}
This internal state keeps track of which packet $U_{s_k}$ should be used by the consistent actuator at time step $k$.
Note that if $\Theta_k=1$ then $s_k=k$ and the latest packet $U_k$ will be used.
Once $s_k$ has been determined, the packet $X_k$ is sent from the plant to the controller with the following content,
\begin{equation}
    \label{eq:Xk}
    X_k = \lbrace \xnom(k), s_k \rbrace.
\end{equation}

While in \cite{RemoteMPCoverLossyNetworks} the packet $X_k$ contains $x(k)$ and $s_k$, our proposed solution sends the nominal state $\xnom(k)$ to the remote controller, which is obtained as described in Section~\ref{sec:NominalPlantAndAncillaryController}. 

The consistent actuator determines $\unom(k)$ as
\begin{align}
    \label{eq:SmartActuator}
    \unom(k) = \begin{cases}
        \mathbf{u}_{s_k}^*(k-s_k) &\mathrm{if}\ k-s_k<N,\\
         \bar{u}_{s_k}^*+\bar{K}\bar{x}_{s_k}^*-\bar{K}\xnom(k) &\mathrm{otherwise}.
    \end{cases}
\end{align}
In a nutshell, the consistent actuator uses all predicted control inputs in a packet $U_k$ if no new consistent packet has been received and once there are no more predicted inputs available it uses the controller $\bar{K}$ in \eqref{eq:nominalsteadystatecontroller} to control the nominal plant around the steady state $\bar{x}_{s_k}^*$.

Note that since $\unom(k)$ is determined from an optimal trajectory coming from the MPC, it is guaranteed that ${\unom(k)\in\Uc}$.

\subsection{Nominal Plant and Ancillary Controller}
\label{sec:NominalPlantAndAncillaryController}
The main idea of our proposed approach is that a model of the nominal plant runs on the local plant to determine the nominal plant state $\xnom(k)$. Here, $\xnom(k)$ together with $\unom(k)$ coming from \eqref{eq:SmartActuator} are used to determine the control input $u(k)$ for the plant via the  ancillary controller $K$ in \eqref{eq:AncillaryController}. 

The nominal control input coming from the consistent actuator is then applied to the model of the nominal plant, which evolves as described in \eqref{eq:NominalPlant}.
Since the nominal control inputs are determined by the MPC problem, they guarantee that $\xnom(k)\in\Xc$ for all $k\geq 0$.

As described in Section~\ref{sec:Preliminaries}, the ancillary controller will guarantee that $x(k)\in \lbrace\xnom(k)\rbrace \oplus \mrpi \subseteq \X$ for all $k\in\mathbb{N}_{\geq 0}$ and $u(k)\in \lbrace \unom(k)\rbrace \oplus K\mrpi \subseteq \U$ if $\xnom(k)\in\Xc$ for all $k\in\mathbb{N}_{\geq 0}$, since $\unom(k)\in\Uc$ for all $k\in \mathbb{N}_{\geq 0}$.

\textit{The nominal plant and ancillary controller on the local plant are the key to make our approach work} because they enable us to track a reference $x_r$ in the presence of a disturbance $w(k)$, and they are the main architectural difference to \cite{RemoteMPCoverLossyNetworks}.
Furthermore, running a nominal model is computationally cheaper than running a robust MPC as in \cite{RobustTubeBasedTrackingMPC} on the local plant. This makes our proposed approach more applicable to lightweight devices controlled over a lossy network without sacrificing robustness.

\begin{remark}
The ancillary controller $K$ and the steady-state controller $\bar{K}$ are not necessarily the same. This enables us to tune $K$ to minimize $\mrpi$, while $\bar{K}$ can be tuned to increase the size of $\Xf$. For the former, Section~7 in \cite{RobustTubeBasedTrackingMPC} proposed a semi-definite program to design $K$, which minimizes $\mrpi$ while guaranteeing that $\Xc$ and $\Uc$ are non-empty.
For the latter, a common choice in the literature is to choose $\bar{K}$ as the optimal LQR gain.
\end{remark}

\subsection{Estimator}
\label{sec:Estimator}
The estimator, similar to \cite{RemoteMPCoverLossyNetworks}, is used to estimate the state of the nominal plant at time step $k+1$ as $\hat{x}(k+1|k)$.
Based on the reception of $X_k$, it estimates the nominal plant state 
\begin{equation}
\label{eq:Estimator}
    \hat{x}(k+1|k) = A\hat{x}(k|k)+B\hat{u}(k|k),
\end{equation}
where 
\begin{align}
    \hat{x}(k|k) &= \gamma_k \xnom(k) + (1-\gamma_k)\hat{x}(k|k-1),\\
    \hat{u}(k|k) &= \gamma_k \unom(k) + (1-\gamma_k)\mathbf{u}_k^*(0).
\end{align}
Since only $\xnom(k)$ and $s_k$ are sent to the remote controller, the remote controller also needs to run a consistent actuator \eqref{eq:SmartActuator} to determine $\unom(k)$.
Furthermore, $q_k$ is updated as follows
\begin{equation}
    q_{k+1}=\gamma_k k + (1-\gamma_k) q_{k}
\end{equation}
to keep track of which packet $X_k$ has been received last at the remote controller.

Other than in \cite{RemoteMPCoverLossyNetworks}, we estimate the nominal plant state in the estimator and not the plant state. 
This guarantees that $\hat{x}(k|k-1)\in \Xc$, such that the constraints $\mathbf{x}(0)=\hat{x}(k|k-1)$ and $\mathbf{x}(0)\in\Xc$ in the optimization problem \eqref{eq:RemoteRobustTubeTrackingMPC} will not lead to an infeasible optimization problem.

\subsection{Theoretical Guarantees}
\label{sec:TheoreticalGuarantees}
In this section, we provide theoretical guarantees for our proposed MPC. The key insight for our theoretical guarantees is that the closed-loop system involving the MPC in Figure~\ref{fig:RemoteTubeTrackingMPCLayout} acts on the nominal plant and not the plant itself. This means that inside this closed-loop system there is no disturbance, such that it represents the disturbance-free system assumed in \cite{RemoteMPCoverLossyNetworks}.
Hence, the theoretical guarantees of \cite{RemoteMPCoverLossyNetworks} will hold for the closed-loop system involving the MPC in our proposed approach given Assumption~\ref{assum:RemoteTrackingMPC} below.
\begin{assumption}
\label{assum:RemoteTrackingMPC}
In addition to Assumptions~\ref{assum:Network} -- \ref{assum:StateAndInputSet}, the following conditions hold:
\begin{enumerate}
\item $Q$, $R$, and $T$ are positive definite. 
\item The system $(Q^{\frac{1}{2}},A)$ is observable.
\item The gains $K$ and $\bar{K}$ are such that $\rho(A-BK)<1$ and $\rho(A-B\bar{K})<1$, respectively. 
\item The matrix $P$ satisfies

$P=(A-B\bar{K})^\top P(A-B\bar{K})+Q+\bar{K}^\top R \bar{K}$.
\end{enumerate}
\end{assumption}

We begin by showing that the plant state is in a bounded neighbourhood around the estimated state if $\Theta_k=1$.
\begin{proposition}
If $\Theta_k=1$, then $x(k)\in \lbrace \hat{x}(k|k-1) \rbrace \oplus \mrpi$.
\end{proposition}
\begin{proof}
Since the closed-loop system of our proposed approach acts on the disturbance-free nominal plant (see Figure~\ref{fig:RemoteTubeTrackingMPCLayout}), we can use Proposition~1 of \cite{RemoteMPCoverLossyNetworks} to show that if $\Theta_k=1$ then $\hat{x}(k|k-1)=\xnom(k)$. Due the ancillary controller, we know that $x(k)\in \lbrace \xnom(k) \rbrace \oplus \mrpi$ holds.
\end{proof}

This shows that when the estimate is consistent with the nominal plant state, i.e. $\Theta_k=1$, then we know that the plant state is in a tube around the estimated state.

Next, we show recursive feasibility of our proposed remote MPC and that the plant will always evolve in the constraints regardless of the network quality.
\begin{proposition}
\label{prop:RecursiveFeasibility}
Let Assumption~\ref{assum:RemoteTrackingMPC} hold, and assume there exists a $k_0$ such that $\gamma_{k_0-1}=1$, $\theta_{k_0}=1$, ${x(k_0)-\xnom(k_0)\in\mrpi}$, and that the optimization problem \eqref{eq:RemoteRobustTubeTrackingMPC} is feasible. 
If the consistent actuator \eqref{eq:SmartActuator} and the ancillary controller \eqref{eq:AncillaryController} are used, the optimization problem \eqref{eq:RemoteRobustTubeTrackingMPC} is feasible, and $x(k)\in\X$ and $u(k)\in\U$ for all $k\geq k_0$.
\end{proposition}
\begin{proof}
Given the conditions above, Proposition~2 of \cite{RemoteMPCoverLossyNetworks} shows us that optimization problem \eqref{eq:RemoteRobustTubeTrackingMPC} is feasible and $\xnom(k)\in\Xc$ and $\unom(k)\in\Uc$ for all $k\geq k_0$. 
The constraint satisfaction of $x(k) \in \X$ and $u(k) \in \U$ is guaranteed since the ancillary controller \eqref{eq:AncillaryController} guarantees that $x(k)\in\lbrace \xnom(k) \rbrace \oplus \mrpi\subseteq\X$ and $u(k)\in\lbrace \unom(k) \rbrace \oplus (-K)\mrpi\subseteq\U$.
\end{proof}

Note that the feasibility of the MPC does not depend on the value of $x_r$, such that for all reference values our solution is recursively feasible according to Proposition~\ref{prop:RecursiveFeasibility}.

Finally, the following theorem states the tracking capabilities of our approach given a constant reference $x_r$.
\begin{theorem}
\label{thm:Tracking}
Let Assumption~\ref{assum:RemoteTrackingMPC} hold and $[x_r^\top,\ \tilde{u}^\top]^\top$ fulfil the steady-state equation \eqref{eq:SteadyStateEquation}.
If the consistent actuator \eqref{eq:SmartActuator} and the ancillary controller \eqref{eq:AncillaryController} are used, then almost surely $\lim_{k\rightarrow\infty}x(k)\in \lbrace \tilde{x}_r \rbrace \oplus \X$, where $\tilde{x}_r=x_r$ if $x_r\in\lambda\Xc$ and $\tilde{u}\in \lambda\Uc$, and $\tilde{x}_r=\arg\min_{x\in \lambda\Xc}\|x-x_r\|_T^2$ otherwise.
\end{theorem}
\begin{proof}
From Proposition~3 of \cite{RemoteMPCoverLossyNetworks} we obtain that $\lim_{k\rightarrow\infty}\xnom(k) =  \tilde{x}_r $ almost surely, while Theorem~1 of \cite{RobustTubeBasedTrackingMPC} states that $\tilde{x}_r=\arg\min_{x\in \lambda\Xc}\|x-x_r\|_T^2$ such that $\tilde{x}_r = x_r$ if $x_r\in\lambda\Xc$ and $\tilde{u}\in \lambda\Uc$.
The ancillary controller guarantees that $\lim_{k\rightarrow\infty}x(k)\in \lbrace \tilde{x}_r \rbrace \oplus \X$ almost surely.
\end{proof}
\begin{corollary}
Theorem~\ref{thm:Tracking} and Proposition~\ref{prop:RecursiveFeasibility} show us that by choosing $f(\cdot)$ and $g(\cdot)$ as in our approach, we have solved Problem~\ref{prob:TrackingOverLossyNetworkWithLocalDisturbance} for constant references.
\end{corollary}

\subsection{Extension to include state feedback}
\label{sec:ExtendedMPCApproach}
While our proposed approach does not require feedback from $x(k)$, it is common to send the state also to the remote controller, for example, for anomaly detection purposes.
Therefore, we will now propose an extension to our approach, which includes state feedback, while inheriting the theoretical guarantees of our previously described approach.

To include the state, we change the content of the plant packet \eqref{eq:Xk} as follows
\begin{equation}
    \label{eq:XkWithState}
    X_k = \lbrace x(k),\xnom(k), s_k \rbrace.
\end{equation}
With the new package \eqref{eq:XkWithState}, the estimator in \eqref{eq:Estimator} uses
\begin{align}
    \hat{x}(k|k) &= \gamma_k x(k) + (1-\gamma_k)\mathbf{x}_k^*(0),\\
    \hat{u}(k|k) &= \gamma_k u(k) + (1-\gamma_k)\mathbf{u}_k^*(0).
\end{align}
Note that we use the state $x(k)$ and control input $u(k)$, when $\gamma_k=1$, where $u(k)$ can be calculated according to \eqref{eq:AncillaryController}. 
This leads to $x(k+1)\subseteq\hat{x}(k+1|k)\oplus \mathbb{W}$, which gives us a better estimate than with the estimator of Section~\ref{sec:Estimator}, where ${x(k+1)\subseteq\hat{x}(k+1|k)\oplus \mrpi}$.
Otherwise, the estimator will use the last optimal trajectory of the MPC to estimate the next state, which gives us again an estimate of the nominal plant.
However, this new estimate does not guarantee that $\hat{x}(k+1|k)\in\Xc$ when $\gamma_k=1$, which requires us to change the constraint \eqref{eq:MPCInitCondition} in our MPC described in Section~\ref{sec:ProposedRemoteMPC} to guarantee feasibility.
Thus, we replace constraint \eqref{eq:MPCInitCondition} with
\begin{align}
    \label{eq:NewMPCInitCondition}
    \lbrace\hat{x}(k-1|k)\rbrace \oplus \mathbb{W}\subseteq \lbrace \mathbf{x}_k(0)\rbrace\oplus \mrpi,
\end{align}
when $\gamma_{k-1}=1$ and otherwise we keep \eqref{eq:MPCInitCondition}.
Hence, the MPC algorithm is now made aware if packets have been received.
Furthermore, the constraint \eqref{eq:NewMPCInitCondition} allows the MPC to reset the nominal state trajectory, since now it is not necessarily true that $\mathbf{x}_k(0)=\hat{x}(k|k-1)$ as it is the case for \eqref{eq:MPCInitCondition}. 
This can improve the convergence as discussed in Chapter 3.5 of \cite{MPCBook}.

Since the MPC can change the optimal trajectory of the nominal plant, we need to update the trajectory on the nominal plant if a consistent packet has been received. 
This is done by changing the controller packet \eqref{eq:Uk} to
\begin{equation}
    \label{eq:UkWithNomState}
    U_k = \lbrace \mathbf{u}_k^*, \bar{u}_k^*+\bar{K}\bar{x}_k^*,\mathbf{x}_k^*(0), q_k\rbrace,
\end{equation}
and setting $\xnom(k)=\mathbf{x}_k^*(0)$ if $\Theta_k=1$.

\begin{proposition}
\label{prop:RecursiveFeasibilityExtendedMPC}
Let Assumption~\ref{assum:RemoteTrackingMPC} hold, and assume there exists a $k_0$ such that $\gamma_{k_0-1}=1$, $\theta_{k_0}=1$, ${x(k_0)-\xnom(k_0)\in\mrpi}$, and that the optimization problem \eqref{eq:RemoteRobustTubeTrackingMPC} is feasible with the new constraint \eqref{eq:NewMPCInitCondition}. 
If the consistent actuator \eqref{eq:SmartActuator} with the nominal state update and the ancillary controller \eqref{eq:AncillaryController} are used, the optimization problem \eqref{eq:RemoteRobustTubeTrackingMPC} with the new constraint \eqref{eq:NewMPCInitCondition}  is feasible, and $x(k)\in\X$ and $u(k)\in\U$ for all $k\geq k_0$.
\end{proposition}
\begin{proof}
If $\gamma_k=0$, the problem is feasible, since the nominal state is used in the estimator. 
If $\gamma_{k}=1$ we can show that
\begin{align}
\hat{x}(k|k-1)\in \lbrace\xnom(k)\rbrace\oplus(A-BK)\mrpi.
\end{align} 
holds. This leads to
\begin{align}
x(k+1)&\in\lbrace\hat{x}(k|k-1)\rbrace \oplus \mathbb{W}
\subseteq \lbrace\xnom(k)\rbrace \oplus \mrpi.
\end{align}
Hence, the constraints $\lbrace\hat{x}(k|k-1)\rbrace \oplus \mathbb{W}\subseteq \lbrace\mathbf{x}_k(0)\rbrace \oplus \mrpi$ and $\mathbf{x}_k(0)\in\Xc$ are feasible with the choice of ${\mathbf{x}_k(0)=\xnom(k)}$.
So the optimal solution of our original MPC \eqref{eq:RemoteRobustTubeTrackingMPC} is a feasible solution of the extended MPC with constraint \eqref{eq:NewMPCInitCondition}. 
Thus, the extended MPC with state feedback is recursively feasible for all $k\geq k_0$, since the original MPC is recursive feasible as shown in Proposition~\ref{prop:RecursiveFeasibility}.
Since $\mathbf{x}_k^*(0)\in\Xc$, the nominal state update, when $\Theta_k=1$, will not change the guarantees given by the ancillary controller, such that $x(k)\in\X$ and $u(k)\in\U$ for all $k\geq k_0$.
\end{proof}
\begin{corollary}
The tracking guarantees of Theorem~\ref{thm:Tracking} hold for the extended MPC with state feedback as well.
\end{corollary}
\begin{proof}
Since Proposition~\ref{prop:RecursiveFeasibilityExtendedMPC} shows that the solution of the original MPC is a feasible solution of the extended MPC, we can deduce that the tracking guarantees of the original MPC also hold for the extended MPC. 
\end{proof}
In summary, this extension includes state feedback from the plant, which can change the optimal trajectory of the nominal plant to improve performance as well with the same theoretical guarantees of the previous approach. However, this approach requires more bandwidth and might change the execution times of the MPC.
\section{Numerical Examples}
\label{sec:simulations}

To demonstrate the efficacy of our proposed approach, henceforth called \emph{\RTMPC{}} and \emph{\ERTMPC{}} for the extended version with state feedback (see Section~\ref{sec:ExtendedMPCApproach}), we use it to track a position reference of a cartpole system, where the pole is in the upright unstable configuration. 
We compare our approach with the approach of \cite{RemoteMPCoverLossyNetworks}, subsequently called \RMPC{}. 
Scripts to reproduce the results presented are included in our open-source code.

In order to design our nominal plant, we linearize the nonlinear dynamics around the unstable equilibrium point, where the pole is pointing up. The resulting continuous-time matrices are defined as follows for the state $x~=~\begin{bmatrix} p & \dot{p} & \phi & \dot{\phi} \end{bmatrix}^\top$:
\begin{equation*}
    A_c = \begin{bmatrix}
        0 & 1 & 0 & 0\\
        0 & \frac{-(I+ml^2)b}{r} & \frac{-m^2gl^2}{r} & 0\\
        0 & 0 & 0 & 1\\
        0 & \frac{-(mlb)}{r} & \frac{mgl(M+m)}{r} & 0
    \end{bmatrix}
    \;
    B_c = \begin{bmatrix}
        0\\
        \frac{I+ml^2}{r}\\
        0\\
        \frac{-ml}{r}
    \end{bmatrix},
\end{equation*}
where $p$ is the position of the cart, $\phi$ the angle of the pole, ${r = I(M+m)+Mml^2}$, with the remaining parameters and their values defined in Table~\ref{tab:parameters}. 
The system is then discretized with a zero-order hold and a sampling time of $T_s=\SI{20}{\milli\second}$ in order to obtain \eqref{eq:NominalPlant}. 
The controllers $K$ and $\bar{K}$ are designed as a discrete LQR controller with cost matrices ${Q=\mathrm{diag}(100, 10, 100, 10)}$ and $R=0.1$.
Furthermore, we choose $|p|\leq \SI{5}{\meter}$, $|\dot{p}|\leq \SI{5}{\meter/\second}$, $|\phi|\leq \SI{0.3}{\radian}$, $|\dot{\phi}|\leq \SI{2}{\radian/\second}$, and $|u|\leq \SI{10}\newton$ to define $\X$ and $\U$. 
The constraints on $\phi$ and $\dot{\phi}$ guarantee that the LQR controller stabilizes the system.
Finally, we choose $N=20$ as the horizon for the MPC.

\subsection{Disturbance set $\mathbb{W}$}
The linearized model will inherently differ from the nonlinear one, and such model error will be represented as the disturbance $w(k)$. 
To estimate the set $\mathbb{W}$, we run several simulations with randomly chosen initial conditions, and let the LQR controller bring the system back to the origin.
The disturbance is then estimated as the difference between the actual state and the linear model, i.e.
$w(k) = x(k+1)-(A-BK)x(k)$. 
This results in the following bounds for the disturbance of the position $|w_p|\leq \SI{0.0001}{\meter}$, velocity $|w_{\dot{p}}|\leq \SI{0.0027}{\meter/\second}$, angle $|w_{\phi}|\leq \SI{0.0003}{\radian}$, and angular velocity $|w_{\dot{\phi}}|\leq \SI{0.043}{\radian/\second}$. 
To approximate $\mrpi$ we use a method described in \cite{RPI_Darup}.

\begin{table}
    \vspace{2mm}
    \centering
    \caption{Parameters used in the numerical examples.}
    \label{tab:parameters}
    \begin{tabular}{ccc}\toprule
         & Definition & Value\\\midrule
        $I$ & Pendulum's inertia & $\SI{0.001}{\kilogram\cdot\meter^2}$ \\
        $l$ & Length to pendulum center of mass & \SI{0.5}{\meter}\\
        $m$ & Pendulum's mass & \SI{0.1}{\kilogram}\\
        $M$ & Cart's mass & \SI{1.0}{\kilogram}\\
        $b$ & Cart's coefficient of friction & \SI{0}{\newton/\meter/\second}\\
        $g$ & Gravity acceleration & \SI{9.8}{\meter/\second^2}\\
        $T_s$ & Sampling Time & \SI{0.02}{\second}\\ \bottomrule
    \end{tabular}
\end{table}

\subsection{Reference Tracking}
Next, we present results for the tracking of a constant reference in position $p$. To do so, the cartpole system is always initialized at the origin, and the reference is set to $r(k)=\begin{bmatrix}0.5, 0, 0, 0\end{bmatrix}^\top$.
To evaluate the performance, we use the average tracking error $\frac{1}{T+1}\sum_{i=0}^T\|x(k)-r(k)\|_2$. 
For the lossy network, we assume a constant packet loss probability of $\varrho$ and investigate $\varrho\in\lbrace0, 0.1, \ldots, 0.9 \rbrace$.
In addition to that, we perform $20$ simulations for each value of $\varrho$ and record the average tracking error to get a better insight for different realizations of the lossy network.

\subsubsection{Linear Plant}
\label{sec:SimulationsWithLinearPlant}
We begin by comparison with the plant being simulated with linear dynamics, where the disturbance $w(k)$ is sampled uniformly from the set $\mathbb{W}$ at each time step.
The results are presented as a box plot in Figure~\ref{fig:BoxplotLinearSystem}.

First, note that both \RTMPC{} and \ERTMPC{} outperform \RMPC{} for every packet loss probability investigated.
Second, for a packet loss probability of $\varrho=0.9$, the average tracking error decreases.
The reason for that is that due to the large packet loss the cartpole moves less aggressively than when there is less packet loss. This leads to a smaller tracking error for the velocity, angle, and angular velocity, since their reference values are zero, which lowers the overall tracking error.
Third, during our simulations, we encountered infeasibility issues for \RMPC{}. While \cite{RemoteMPCoverLossyNetworks} proves recursive feasibility for the plant without a disturbance present, the presence of a disturbance in our simulations showed that infeasibility can occur.
Hence, modelling errors can result in infeasible MPC problems for \RMPC{}, which we will encounter again when the nonlinear plant is used.
Comparing \RTMPC{} and \ERTMPC{} we observe that the performance of \ERTMPC{} seems almost constant, while the tracking error for \RTMPC{} increases with the packet loss probability. The ability to reset the nominal trajectory is likely the reason for the constant performance of \ERTMPC{}.

\subsubsection{Nonlinear Plant}
Next, we compare the controllers on the nonlinear cartpole simulated using PyBullet. 
To do so, the physics simulators runs at a higher frequency than the controllers (\SI{500}{\hertz}, to be precise), and a zero-order hold keeps the control input constant between controller updates.

Figure~\ref{fig:Boxplot} shows the box plots of the average tracking error for the different packet loss probabilities. 
We observed that \RMPC{} struggled with infeasibility issues; notably, for $\varrho\in\lbrace 0, 0.1, 0.2, 0.3, 0.4\rbrace$ \RMPC{} is always infeasible in our simulations, and the larger $\varrho$ the less infeasible problems were encountered.  
Therefore, the corresponding box plots only present the results of runs without an infeasible MPC problem. 
Our approach, on the other hand, is recursively feasible for all simulations performed.

The infeasibility issues decreasing with the increase of packet losses for \RMPC{} was a surprising result, since the opposite could sound more logical.
Our intuition for this is that the LQR controller $\bar{K}$ used as the steady state controller is able to handle the nonlinearities of the system better than \RMPC{}, since it uses direct state feedback, while \RMPC{} estimates the next state based on the currently received state.
Hence, the more packet loss there is, the more often the steady-state controller is used, which brings the plant to a state that \RMPC{} can actually handle well.
Our approach, on the other hand, uses the LQR controller both as the steady-state controller in the MPC as well as the tracking controller to track the nominal plant state and, in addition to that, tightens the constraint set of the MPC by taking the propagation of the modelling error into account.
This can be observed in Figure~\ref{fig:Trajectories}, where we present one trajectory of the position and angle at a packet loss of $\varrho=0.4$, and the star marks when the infeasibility occurred in \RMPC{}.
\RMPC{} exhibits an oscillatory behaviour before it becomes infeasible, while \RTMPC{} has a smoother trajectory, which reaches the desired reference. By including actual state feedback in \ERTMPC{} the trajectory becomes even smoother due to the ability to reset the nominal trajectory based on the state $x(k)$.

While our approaches have not shown any infeasibility issues, we noticed that the state is not always in a tube around the nominal state for \ERTMPC{}. 
These violations happened in the beginning of the simulation and then stopped. 
We believe that $\mathbb{W}$ does not capture the differences well in the beginning of the reference tracking which leads to these violations. We did not observe such violations for \RTMPC{}, probably because it is more conservative than \ERTMPC{}.

In general, we observe that our proposed solution outperforms \RMPC{} of \cite{RemoteMPCoverLossyNetworks} for all investigated value of $\varrho$. 
Interestingly, the tracking error seems to peak at $\varrho=0.7$ and then reduces again for \RMPC{} and \RTMPC{}, which is due to the same reason as in the linear case.

\subsection{Execution time of the MPC}
Our simulations run on a 24GB RAM Windows machine with a Ryzen7 8-core CPU. From the 50000 executions of the MPC in Section~\ref{sec:SimulationsWithLinearPlant}, we removed the first execution time, since it represents the cold start of the optimization, and present the histogram of the remaining execution times in Figure~\ref{fig:HistogramExecutionTimes}.
We observe that the majority of the sampling times is below \SI{20}{\milli\second}, which shows that our MPC can run in real-time for the sampling time of \SI{20}{\milli\second}.
Further, the median and the $95\,\%$ quantile of the execution time for \RTMPC{} were \SI{5.00}{\milli\second} and \SI{7.85}{\milli\second}, respectively. 
The median and the $95\,\%$ quantile of the execution time for \ERTMPC{} were \SI{6.21}{\milli\second} and \SI{7.06}{\milli\second}, respectively. 
The histogram for \ERTMPC{} has two peaks because it solves two different MPC problems depending on if a measurement was received or not.
While real-time execution is not considered here, an optimization problem that is not solved in time can be interpreted as a lost packet in a real scenario. Hence, our approach can deal with too long execution times of the MPC as well.

\begin{figure}
    \centering
    \subfloat[Linear plant]{
        \includegraphics[trim=8 4 45 10,clip,width=0.8\linewidth]{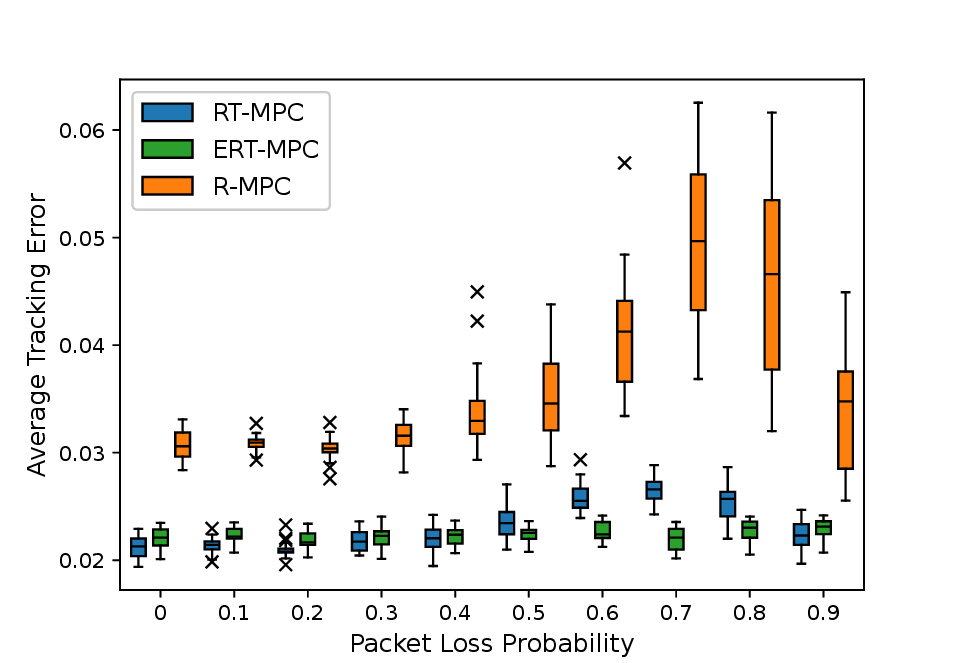}
        \label{fig:BoxplotLinearSystem}
    }
    \\
    \subfloat[Nonlinear plant]{
        \includegraphics[trim=8 4 45 10,clip,width=0.8\linewidth]{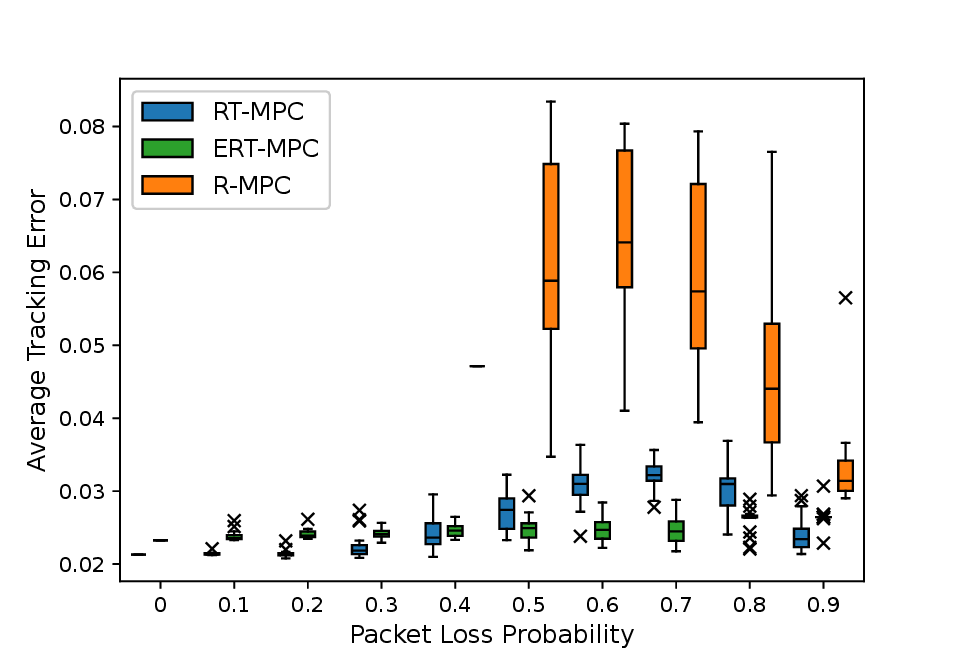}
        \label{fig:Boxplot}
    }
    \\
    \subfloat[Sample trajectory]{
        \includegraphics[trim=10 2 10 5,clip,width=0.8\linewidth]{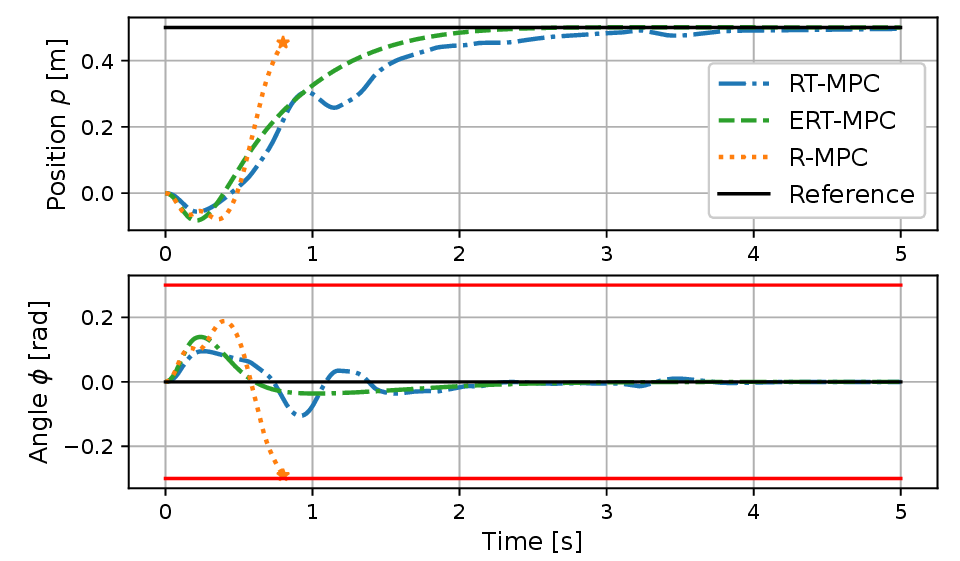}
        \label{fig:Trajectories}
    }
    \\
    \subfloat[MPC execution time]{
        \includegraphics[trim=4 0 46 25,clip,width=0.8\linewidth]{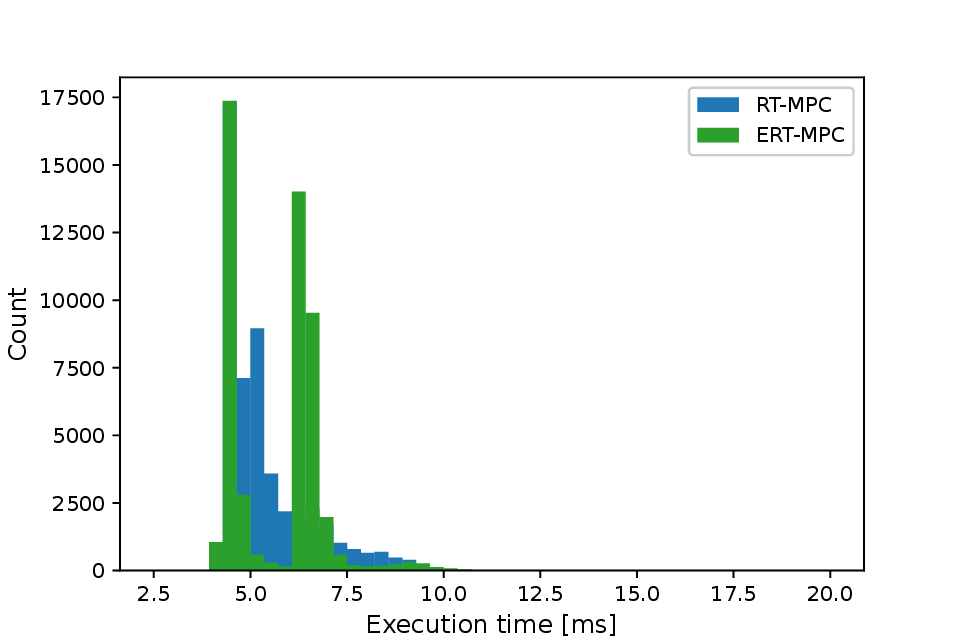}
        \label{fig:HistogramExecutionTimes}
    }
    \caption{Results comparing our approaches, \RTMPC{} and \ERTMPC{}, and \RMPC{} of \cite{RemoteMPCoverLossyNetworks}. In (a) and (b), the average tracking error over 20 runs for different packet loss probabilities is presented. Then, an example trajectory for a packet loss probability of $\varrho=0.4$ is presented in (c), where the star marks the moment when \RMPC{} becomes infeasible. Lastly, (d) shows the histogram of the execution time of the MPC.}
    \label{fig:results}
\end{figure}
\section{Conclusions}
\label{sec:conclusions}

We presented a novel framework that addresses the problem of controlling systems over lossy network connections.
More precisely, we propose a robust tube-based MPC algorithm that allows for the tracking of a piecewise-constant reference signal with guaranteed convergence properties for constant references, recursive feasibility, and safety and input constraint satisfaction.
Further, we presented numerical simulation results of the approach applied to a cartpole system, together with comparisons with state-of-the-art algorithms.
Lastly, our code is available as open-source. 

For future work, we would like to investigate time-varying trajectories and the reasons for the peak of the reference tracking error around a packet loss probability of $80\%$.



\bibliographystyle{IEEEtran}
\bibliography{biblio}

\end{document}